\documentclass[11pt]{article}

\def\colorful{0}

\usepackage{amsthm,amsfonts,amsmath,amssymb,epsfig,color,float,graphicx,verbatim}
\usepackage{multirow}

\newif\ifhyper\IfFileExists{hyperref.sty}{\hypertrue}{\hyperfalse}
\hypertrue
\ifhyper\usepackage{hyperref}\fi

\usepackage{enumitem}
\usepackage[capitalize]{cleveref} 

\makeatletter
\renewcommand{\subsubsection}{\@startsection{subsubsection}{3}{0pt}{-12pt}{-5pt}{\normalsize\bf}}
\makeatother

\usepackage{framed}
\usepackage{nicefrac}

\def\nnewcolor{1}
\ifnum\nnewcolor=1

\fi
\ifnum\nnewcolor=0

\fi

\ifnum\colorful=1
\newcommand{\new}[1]{{\color{red} #1}}

\else
\newcommand{\new}[1]{{#1}}

\fi

\newtheorem{theorem}{Theorem}

\newtheorem{lemma}[theorem]{Lemma}
\newtheorem{proposition}[theorem]{Proposition}
\newtheorem{corollary}[theorem]{Corollary}

\theoremstyle{definition}
\newtheorem{definition}[theorem]{Definition}

\newcommand{\R}{\mathbb{R}}
\newcommand{\Z}{\mathbb{Z}}
\newcommand{\E}{\mathbb{E}}

\newcommand{\dk}{d_{\mathrm K}}

\newcommand{\var}{\text{Var}}

\newcommand{\depth}{\mathbf{depth}}
\newcommand{\width}{\mathbf{width}}
\newcommand{\discr}{\mathbf{Discr}}

\newcommand{\ignore}[1]{}

\newcommand{\eps}{\varepsilon}

\newcommand{\Var}{\mathop{\textnormal{Var}}\nolimits}

\newcommand{\Poi}{\mathop{\textnormal{Poi}}\nolimits}

\renewcommand{\eqref}[1]{Eq.~(\ref{#1})}

\newcommand{\eqdef}{\stackrel{{\mathrm {\footnotesize def}}}{=}}

\newenvironment{algorithm}[1][\  ] %
{ \rm
\begin{tabbing}
....\=.....\=.....\=.....\=.....\=  \+ \kill
} %
{\end{tabbing} }

\floatstyle{ruled}
\newfloat{Algorithm}{H}{alg}
\newfloat{Subroutine}{H}{sub}

{
\begin{minipage}{1.0\linewidth} \begin{algorithm} %
} { \end{algorithm} \end{minipage} }

\title{Testing Identity of Structured Distributions}

\author{
Ilias Diakonikolas\thanks{Supported by EPSRC grant EP/L021749/1, a Marie Curie Career Integration Grant, and a SICSA grant.}\\
University of Edinburgh\\
{\tt ilias.d@ed.ac.uk}.\\
\and
Daniel M. Kane\thanks{Supported in part by an NSF Postdoctoral Fellowship.}\\
University of California, San Diego\\
{\tt dakane@cs.ucsd.edu}.\\
\and
Vladimir Nikishkin\thanks{Supported by a University of Edinburgh PCD Scholarship.}\\
University of Edinburgh\\
{\tt v.nikishkin@sms.ed.ac.uk}.
}

\begin{document}

\maketitle

\begin{abstract}
We study the question of identity testing for structured distributions.
More precisely, given samples from a {\em structured} distribution $q$ over $[n]$ and an explicit distribution $p$ over $[n]$,
we wish to distinguish whether $q=p$ versus $q$ is at least $\eps$-far from $p$,
in $L_1$  distance. In this work, we present a unified approach that yields new, simple testers, with sample complexity
that is information-theoretically optimal, for broad classes of structured distributions, including $t$-flat distributions,
$t$-modal distributions, log-concave distributions, monotone hazard rate (MHR) distributions, and mixtures thereof.
\end{abstract}



\section{Introduction}  \label{sec:intro}

How many samples do we need to verify the identity of a distribution?
This is arguably {\em the} single most fundamental question in statistical hypothesis testing~\cite{NeymanP}, with
Pearson's chi-squared test~\cite{Pearson1900} (and variants thereof) still being the method of choice used in practice.
This question has also been extensively studied by the TCS community
in the framework of {\em property testing}~\cite{RS96, GGR98}:
Given sample access to an unknown distribution $q$ over a finite domain $[n]: = \{1, \ldots, n\}$, an explicit distribution
$p$ over $[n]$, and a parameter $\eps>0$, we want to distinguish between the cases that $q$ and $p$ are identical versus
$\eps$-far from each other in $L_1$ norm (statistical distance).
Previous work on this problem focused on characterizing the sample size needed to test the identity of an arbitrary
distribution of a given support size. After more than a decade of study, this ``worst-case'' regime
is well-understood: there exists a computationally efficient estimator with sample complexity $O(\sqrt{n}/\eps^2)$~\cite{VV14}
and a matching information-theoretic lower bound~\cite{Paninski:08}.

While it is certainly a significant improvement over naive approaches and is tight in general,
the bound of $\Theta(\sqrt{n})$ is still impractical, if the support size $n$ is very large.
We emphasize that the aforementioned sample complexity characterizes worst-case instances,
and one might hope that drastically better results can be obtained for most natural settings.
In contrast to this setting, in which we assume
nothing about the structure of the unknown distribution $q$, in many cases we know a priori
that the distribution $q$ in question has some ``nice structure''.
For example, we may have some qualitative information about the density $q$, e.g., it may be a mixture of a small number
of log-concave distributions, or a multi-modal distribution with a bounded number of modes. The following question
naturally arises:
{\em Can we exploit the underlying structure in order to perform the desired statistical estimation task more efficiently?}

One would optimistically hope for the answer to the above question to be ``YES.''
While this has been confirmed in several cases for the problem of {\em learning} 
(see e.g., ~\cite{DDS12soda, DDS12stoc, DDOST13focs, CDSS14}),
relatively little work has been done for {\em testing} properties of structured distributions.
In this paper, we show that this is indeed the case for the aforementioned problem of identity testing
for a broad spectrum of natural and well-studied distribution classes. To describe our results in more detail, we will need
some terminology.

Let $C$ be a class of distributions over $[n]$. The problem of {\em identity testing for $C$} is the following:
Given sample access to an unknown distribution $q \in C$, and an explicit distribution $p \in C$\footnote{It is no loss of generality to assume that $p \in C$; otherwise the tester can output ``NO'' without drawing samples.},
we want to distinguish between the case that $q = p$ versus $\|q-p\|_1 \ge \eps.$ We emphasize that the sample complexity of this testing
problem depends on the underlying class
$C$, and we believe it is of fundamental interest to obtain efficient algorithms that are {\em sample optimal} for $C$.
One approach to solve this problem is to learn $q$ up to $L_1$ distance $\eps/2$ and
check that the hypothesis is $\eps/2$-close to $p$. Thus, the sample complexity of identity testing for $C$
is bounded from above by the sample complexity of {\em learning} (an arbitrary distribution in) $C$.
It is natural to ask whether a better sample size bound could be achieved for the identity testing
problem, since this task is, in some sense,  less demanding than the task of learning.

In this work, we provide a comprehensive picture of the sample and computational complexities of identity testing for a broad class of structured distributions.
More specifically, we propose a unified framework that yields new, simple, and {\em provably optimal} identity testers for various structured classes $C$; see Table~\ref{table:results} for an indicative list of distribution classes to which our framework applies. {\em Our approach relies on a single unified algorithm that we design, which yields
highly efficient identity testers for many shape restricted classes of distributions.}

As an interesting byproduct, we establish that, for various structured classes $C$, identity testing for $C$ is provably easier than learning. In particular, the sample bounds in the third column of Table~\ref{table:results} from~\cite{CDSS14} also apply for
{\em learning} the corresponding class $C$, and are known to be information-theoretically optimal for the learning problem.

Our main result (see Theorem~\ref{thm:main} and Proposition~\ref{prop:simple} in Section~\ref{sec:results}) can be phrased, roughly,  as follows:
{\em Let $C$ be a class of univariate distributions such that any pair of distributions $p, q \in C$ have ``essentially'' at most $k$ crossings, that is,
points of the domain where $q-p$ changes its sign. Then, the identity problem for $C$ can be solved with $O(\sqrt{k}/\eps^2)$ samples.
Moreover, this bound is information-theoretically optimal.}

By the term ``essentially'' we mean that a constant fraction of the contribution to $\|q-p\|_1$ is due
to a set of $k$ crossings -- the actual number of crossings can be arbitrary.
For example, if $C$ is the class of $t$-piecewise constant distributions, it is clear that any two distributions in $C$ have $O(t)$ crossings,
which gives us the first line of Table~\ref{table:results}. As a more interesting example, consider the class $C$ of log-concave distributions over $[n]$.
While the number of crossings between $p, q \in C$ can be $\Omega(n)$, it can be shown (see Lemma~17 in~\cite{CDSS14}) that the essential
number of crossings is $k = \widetilde{O} (1/\sqrt{\eps})$, which gives us the third line of the table. More generally,
we obtain asymptotic improvements over the standard $O(\sqrt{n}/\eps^2)$ bound for any class $C$ such that
the essential number of crossings is $k = o(n)$. This condition applies for any class $C$ that can be well-approximated
in $L_1$ distance by piecewise low-degree polynomials (see Corollary~\ref{cor:approx} for a precise statement).


\begin{table*}[t] \label{table:results}
\begin{center}
\begin{tabular}{|c|c|c|c|}%
\hline \bf Class of Distributions over $[n]$ & \bf Our upper bound & \bf Previous work
\\\hline\hline

$t$-piecewise constant & $O(\sqrt{t}/\eps^2)$ & $O(t/\eps^2)$ \cite{CDSS14} \\ \hline

$t$-piecewise degree-$d$ polynomial & $O\left(\sqrt{t(d+1)}/\eps^2\right)$ &  $O\left(t(d+1)/\eps^2\right)$  \cite{CDSS14}
\\\hline

log-concave & $\widetilde{O}(1/\eps^{9/4})$ & $\widetilde{O}(1/\eps^{5/2})$ \cite{CDSS14} \\ \hline

$k$-mixture of log-concave & $\sqrt{k} \cdot \widetilde{O}(1/\eps^{9/4})$  & $\widetilde{O}(k/\eps^{5/2})$ \cite{CDSS14}  \\ \hline

$t$-modal & $O(\sqrt{t \log(n)}/\eps^{5/2})$ &  $O\left(\sqrt{t \log(n)}/\eps^{3}+t^2/\eps^4\right)$ \cite{DDSVV13}   \\\hline

$k$-mixture of $t$-modal  & $O(\sqrt{k t \log(n)}/\eps^{5/2})$  &  $O\left(\sqrt{k t \log(n)}/\eps^{3}+k^2t^2/\eps^4\right)$ \cite{DDSVV13}  \\\hline

monotone hazard rate (MHR) & $O(\sqrt{\log(n/\eps)}/\eps^{5/2})$ &  $O(\log(n/\eps)/\eps^{3})$  \cite{CDSS14} \\\hline

$k$-mixture of MHR & $O(\sqrt{k \log(n/\eps)}/\eps^{5/2})$ &  $O(k \log(n/\eps)/\eps^{3})$  \cite{CDSS14} \\\hline
\end{tabular}
\end{center}
\vspace{-0.2cm}
\caption{Algorithmic results for identity testing of various
classes of probability distributions.
The second column indicates the sample complexity of our general algorithm applied to the class under consideration.
The third column indicates the sample complexity of the best previously known algorithm for the same problem.
}
\label{tab:results}
\vspace{-0.3cm}
\end{table*}

\subsection{Related and Prior Work} \label{ssec:literature} In this subsection we review the related literature and compare
our results with previous work.

\smallskip

\noindent {\bf Distribution Property Testing} The area of distribution property testing, initiated in the TCS community by the work of Batu {\em et al.}~\cite{BFR+:00, Batu13},
has developed into a very active research area with intimate connections to information theory, learning and statistics.
The paradigmatic algorithmic problem in this area is the following: given sample access to an unknown distribution $q$ over an
$n$-element set, we want to determine whether $q$ has some property or is ``far''
(in statistical distance or, equivalently, $L_1$ norm) from any distribution having the property.
The overarching goal is to obtain a computationally efficient algorithm that uses as few samples as possible --
certainly asymptotically fewer than the support size $n$, and ideally much less than that.
See~\cite{GR00, BFR+:00, BFFKRW:01, Batu01, BDKR:02, BKR:04,  Paninski:08, PV11sicomp, ValiantValiant:11, DDSVV13, DJOP11, LRR11, ILR12} for a sample of works
and~\cite{Rub12} for a survey.

One of the first problems studied in this line of work is that of ``identity testing against a known distribution'': Given samples from an unknown distribution
$q$ and an explicitly given distribution $p$ distinguish between the case that $q = p$ versus the case that $q$ is $\eps$-far from $p$ in $L_1$ norm.
The problem of {\em uniformity testing} -- the special case of identity testing when $p$ is the uniform distribution -- was first considered by Goldreich and Ron~\cite{GR00} who,
motivated by a connection to testing expansion in graphs, obtained a uniformity tester using $O(\sqrt{n}/\eps^4)$ samples.
Subsequently, Paninski gave the tight bound of $\Theta(\sqrt{n}/\eps^2)$~\cite{Paninski:08} for this problem.
Batu {\em et al.}~\cite{BFFKRW:01} obtained an identity testing algorithm against an arbitrary explicit distribution
with sample complexity $\tilde{O}(\sqrt{n}/\eps^4)$. The tight bound of $\Theta(\sqrt{n}/\eps^2)$ for the general identity testing problem was given only recently in~\cite{VV14}.

\smallskip

\noindent {\bf Shape Restricted Statistical Estimation} The area of inference under shape constraints --  that is, inference about a probability distribution
under the constraint that its probability density function (pdf) satisfies certain qualitative properties --
is a classical topic in statistics starting with the pioneering work of Grenander~\cite{Grenander:56}
on monotone distributions (see \cite{BBBB:72} for an early book on the topic).
Various structural restrictions have been studied in the statistics literature, starting from
monotonicity, unimodality, and concavity~\cite{Grenander:56, Brunk:58, PrakasaRao:69, Wegman:70, HansonP:76, Groeneboom:85, Birge:87, Birge:87b,
Fougeres:97,ChanTong:04,JW:09},
and more recently focusing on structural restrictions such as log-concavity and $k$-monotonicity
\cite{BW07aos, DumbgenRufibach:09, BRW:09aos, GW09sc, BW10sn, KoenkerM:10aos}.

Shape restricted inference is well-motivated in its own right, and
has seen a recent surge of research activity in the statistics community, in part due to the ubiquity of structured distributions
in the natural sciences. Such structural constraints on the underlying distributions
are sometimes direct consequences of the studied application problem (see e.g., Hampel~\cite{Hampel87},
or Wang {\em et al.}~\cite{Wang05}), or they are a plausible explanation of the model under investigation (see e.g., ~\cite{Reb05aos} and references therein for
applications to economics and reliability theory). We also point the reader to the recent survey
~\cite{Walther09} highlighting the importance of log-concavity in statistical inference.
The hope is that, under such structural constraints,  the quality of the resulting estimators may dramatically improve, both in terms
of sample size and in terms of computational efficiency.

We remark that the statistics literature on the topic has focused primarily on the problem of {\em density estimation} or learning an unknown structured distribution. That is,
given samples from a distribution $q$ promised to belong to some distribution class $C$, we would like to output
a hypothesis distribution that is a good approximation to $q$. In recent years, there has been a flurry of results in the TCS community
on learning structured distributions, with a focus on both sample complexity and computational complexity, see~\cite{KMR+:94, FOS:05focs, BelkinSinha:10, KMV:10, MoitraValiant:10,DDS12soda,DDS12stoc, CDSS13, DDOST13focs, CDSS14} for some representative works.

\smallskip

\noindent {\bf Comparison with Prior Work}
In recent work, Chan, Diakonikolas, Servedio, and Sun~\cite{CDSS14} proposed a general approach to {\em learn} univariate probability distributions
that are well approximated by piecewise polynomials. \cite{CDSS14} obtained a computationally efficient and sample near-optimal
algorithm to agnostically learn piecewise polynomial distributions, thus obtaining efficient
estimators for various classes of structured distributions. For many of the classes $C$ considered in Table~\ref{table:results}
the best previously known sample complexity for the identity testing problem for $C$ is identified with the sample complexity of the
corresponding learning problem from~\cite{CDSS14}. We remark
that the results of this paper apply to all classes $C$ considered in~\cite{CDSS14}, and are in fact more general
as our condition (any $p, q \in C$ have a bounded number of ``essential'' crossings) subsumes the piecewise polynomial condition (see discussion before
Corollary~\ref{cor:approx} in Section~\ref{sec:results}).
At the technical level,  in contrast to the learning algorithm of~\cite{CDSS14},
which relies on a combination of linear programming and dynamic programming,
our identity tester is simple and combinatorial.

In the context of property testing,
Batu, Kumar, and Rubinfeld~\cite{BKR:04} gave algorithms for the problem of identity testing of unimodal distributions with sample complexity $O(\log^3 n).$
More recently, Daskalakis, Diakonikolas, Servedio, Valiant, and Valiant~\cite{DDSVV13}
generalized this result to $t$-modal distributions obtaining an identity tester  with
sample complexity $O(\sqrt{t \log(n)}/\eps^{3}+t^2/\eps^4)$. We remark that for the class of $t$-modal distributions our approach yields an identity tester
with sample complexity $O(\sqrt{t \log(n)}/\eps^{5/2})$, matching the lower bound of~\cite{DDSVV13}.
Moreover, our work yields sample optimal identity testing algorithms
not only for $t$-modal distributions, but for a broad spectrum of structured distributions via a unified approach.

It should be emphasized that the main ideas underlying this paper are very different from those of ~\cite{DDSVV13}.
The algorithm of~\cite{DDSVV13} is based on the fact from~\cite{Birge:87} that
any $t$-modal distribution is $\eps$-close in $L_1$ norm to a piecewise constant distribution
with $k = O(t \cdot \log(n)/\eps)$ intervals. Hence, if the location and the width of these $k$ ``flat'' intervals
were known in advance, the problem would be easy: The algorithm could just test identity between the ``reduced'' distributions supported on these $k$ intervals,
thus obtaining the optimal sample complexity of $O(\sqrt{k}/\eps^2) = O(\sqrt{t \log (n)}/\eps^{5/2})$.
To circumvent the problem that this decomposition is not known a priori, \cite{DDSVV13}
start by drawing samples from the unknown distribution $q$ to construct such a decomposition.
There are two caveats with this strategy: First, the number of samples used to achieve this is
 $\Omega(t^2)$ and the number of intervals of the constructed decomposition
is significantly larger than $k$, namely $k' = \Omega(k/\eps)$.
As a consequence, the sample complexity of identity testing for the reduced distributions on support $k'$ is $\Omega(\sqrt{k'}/\eps^2) =   \Omega(\sqrt{t \log (n)}/\eps^{3}).$

In conclusion, the approach of~\cite{DDSVV13} involves constructing an {\em adaptive} interval decomposition of the domain
followed by a single application of an identity tester to the reduced distributions over those intervals. At a high-level our novel approach works as follows:
We consider {\em several oblivious} interval decompositions of the domain (i.e., without drawing any samples from $q$)
and apply a ``reduced'' identity tester for {\em each} such decomposition. While it may seem surprising that such an approach can be optimal,
our algorithm and its analysis exploit a certain strong property of uniformity testers, namely their performance guarantee with respect
to the {\em $L_2$ norm}. See Section~\ref{sec:results} for a detailed explanation of our techniques.



\new{Finally, we comment on the relation of this work
to the recent paper~\cite{VV14}.  In~\cite{VV14}, Valiant and Valiant study the sample complexity of the identity testing problem as a function of the explicit distribution.
In particular, \cite{VV14} makes no assumptions about the structure of the unknown distribution $q$, and characterizes the sample complexity
of the identity testing problem as a function of the known distribution $p.$
The current work provides a unified framework to exploit structural properties of the unknown distribution $q$,
and yields sample optimal identity testers for various shape restrictions.
Hence, the results of this paper are orthogonal to the results of \cite{VV14}.
}



\section{Our Results and Techniques} \label{sec:results}

\subsection{Basic Definitions} \label{ssec:defs} We start with some notation that will be used throughout this paper.
We consider discrete probability distributions over $[n]: = \{1, \ldots, n\}$,
which are given by probability density functions $p: [n] \rightarrow [0,1]$ such that $\sum_{i=1}^n p_i =1$, where $p_i$ is the probability of element
$i$ in distribution $p$. By abuse of notation, we will sometimes use $p$ to denote the distribution with density function $p_i$.
\new{We emphasize that we view the domain $[n]$ as an ordered set. Throughout this paper we will be interested in structured distribution families
that respect this ordering.}

The $L_1$ (resp. $L_2$) norm of a distribution is identified with the $L_1$ (resp. $L_2$) norm of the corresponding \new{density function}, i.e.,
$\|p\|_1 = \sum_{i=1}^n |p_i|$ and $\|p\|_2 = \sqrt{\sum_{i=1}^n p^2_i}$. The $L_1$ (resp. $L_2$) distance between distributions $p$ and $q$
is defined as the $L_1$ (resp. $L_2$) norm of the vector of their difference, i.e., $\|p-q\|_1 = \sum_{i=1}^n |p_i -q_i|$ and
$\|p-q\|_2 = \sqrt{\sum_{i=1}^n (p_i-q_i)^2}$. We will denote by $U_n$ the uniform distribution over $[n]$.

\smallskip

\noindent{\bf Interval partitions and $\mathcal{A}_k$-distance}
Fix a partition of $[n]$ into disjoint intervals
$\mathcal{I} :=  (I_i)_{i=1}^{\ell}.$ For such a collection $\mathcal{I}$ we will denote its cardinality by
$|\mathcal{I}|$, i.e., $|\mathcal{I}| =  \ell.$ For an interval $J\subseteq [n]$, we denote by $|J|$ its cardinality or length, i.e.,
if $J = [a, b]$, with $a \le b \in [n]$, then $|J| = b-a+1.$
The {\em reduced distribution} $p_r^{\mathcal{I}}$ corresponding to $p$ and $\mathcal{I}$ is the distribution over $[\ell]$
that assigns the $i$th ``point'' the mass that $p$ assigns to the
interval $I_i$; i.e., for $i \in [\ell]$, $p_r^{\mathcal{I}} (i) = p(I_i)$.

We now define a distance metric between distributions that will be crucial for this paper.
Let $\mathfrak{J}_k$ be the collection
of all partitions of $[n]$ into $k$ intervals, i.e., $\mathcal{I} \in \mathfrak{J}_k$ if and only if $\mathcal{I} =  (I_i)_{i=1}^{k}$
is a partition of $[n]$ into intervals $I_1, \ldots, I_k$. For $p, q: [n] \to [0,1]$ and $k \in \Z_+$, $2\le k \le n$, we define
the $\mathcal{A}_k$-distance between $p$ and $q$ by
$$\|p-q\|_{\mathcal{A}_k} \eqdef \max_{\mathcal{I} = (I_i)_{i=1}^{k} \in \mathfrak{J}_k} \sum_{i=1}^k |p(I_i) - q(I_i)|
= \max_{\mathcal{I}  \in \mathfrak{J}_k} \| p_r^{\mathcal{I} } - q_r^{\mathcal{I} } \|_1.$$
We remark that the $\mathcal{A}_k$-distance between distributions\footnote{We note that the definition of $\mathcal{A}_k$-distance
in this work is slightly different than~\cite{DL:01, CDSS14}, but is easily seen to
be essentially equivalent. In particular, \cite{CDSS14} considers the quantity $\max_{S \in \mathcal{S}_k} |p(S) - q(S)|$,
where  $\mathcal{S}_k$ is the collection of all unions of at most $k$ intervals in $[n]$. It is a simple exercise to verify that
$\|p-q\|_{\mathcal{A}_k} \le 2 \cdot \max_{S \in \mathcal{S}_k} |p(S) - q(S)| \new{=} \|p-q\|_{\mathcal{A}_{2k+1}}$,
which implies that the two definitions are equivalent up to constant factors for the purpose of both upper and lower bounds.}
is well-studied in probability theory and statistics.
Note that for any pair of distributions $p, q:[n] \to [0,1]$, and any $k \in \Z_+$ with $2\le k \le n$,
we have that $\|p-q\|_{\mathcal{A}_k} \le \|p-q\|_1$, and the two metrics are identical for $k=n$.
Also note that $\|p-q\|_{\mathcal{A}_2} = 2 \cdot \dk(p, q)$, where $\dk$ is the Kolmogorov metric (i.e., the $L_{\infty}$ distance between the CDF's).

\smallskip

\noindent {\bf Discussion} The well-known \emph{Vapnik-Chervonenkis (VC) inequality}
(see e.g., {\cite[p.31]{DL:01}}) provides the information-theoretically optimal sample size
to {\em learn} an arbitrary distribution $q$ over $[n]$ in this metric. In particular,
it implies that $m = \Omega(k/\eps^2)$ iid draws from $q$
suffice in order to learn $q$ within $\mathcal{A}_k$-distance $\eps$ (with probability at least $9/10$).
This fact has recently proved useful in the context of learning structured distributions: By exploiting this fact,
Chan, Diakonikolas, Servedio, and Sun~\cite{CDSS14} recently obtained computationally efficient and near-sample optimal algorithms for learning
various classes of structured distributions {\em with respect to the $L_1$ distance}.

It is thus natural to ask the following question: What is the sample complexity of {\em testing} properties of distributions
with respect to the $\mathcal{A}_k$-distance? Can we use property testing algorithms in this metric
to obtain sample-optimal testing algorithms for interesting classes of structured distributions  {\em with respect to the $L_1$ distance}?
In this work we answer both questions in the affirmative
for the problem of identity testing.

\subsection{Our Results} \label{ssec:results}

Our main result is an optimal algorithm for the identity testing problem under the
$\mathcal{A}_k$-distance metric:
\begin{theorem}[Main] \label{thm:main}
Given $\eps>0$, an integer $k$ with $2 \le k \le n$, sample access to a distribution $q$ over $[n]$,
and an explicit distribution $p$ over $[n]$,
there is a computationally efficient algorithm which uses $O(\sqrt{k}/\eps^2)$ samples from $q$,
and with probability at least $2/3$ distinguishes whether $q = p$ versus $\|q-p\|_{{\cal A}_k} \ge \eps$.
Additionally, $\Omega(\sqrt{k}/\eps^2)$ samples are information-theoretically necessary.
\end{theorem}

The information-theoretic sample lower bound of $\Omega(\sqrt{k}/\eps^2)$ can be easily deduced from the known lower bound of
$\Omega(\sqrt{n}/\eps^2)$ for uniformity testing over $[n]$ under the $L_1$ norm~\cite{Paninski:08}.
Indeed, if the underlying distribution $q$ over $[n]$ is piecewise constant with $k$ pieces, and $p$ is the uniform distribution over $[n]$,
we have $\|q-p\|_{\mathcal{A}_k} = \|q-p\|_1.$ Hence, our $\mathcal{A}_k$-uniformity testing problem in this case is at least as hard
as $L_1$-uniformity testing over support of size $k$.

\new{The proof of Theorem~\ref{thm:main} proceeds in two stages: In the first stage, we reduce the $\mathcal{A}_k$ identity testing problem
to $\mathcal{A}_k$ uniformity testing without incurring any loss in the sample complexity.
In the second stage, we use an optimal $L_2$ {\em uniformity} tester as a black-box to obtain an $O(\sqrt{k}/\eps^2)$ sample algorithm
for $\mathcal{A}_k$ uniformity testing. We remark that the $L_2$ uniformity tester is not applied to the distribution $q$ directly, but to a sequence
of reduced distributions $q_r^{\cal{I}}$, for an appropriate collection of interval partitions $\cal{I}$. See Section~\ref{ssec:techniques} for a detailed
intuitive explanation of the proof.

We remark that an application of Theorem~\ref{thm:main} for $k=n$, yields
a sample optimal $L_1$ identity tester (for an arbitrary distribution $q$),
giving a new algorithm matching the recent tight upper bound in~\cite{VV14}. Our new $L_1$ identity tester is arguable simpler and more intuitive,
as it only uses an $L_2$ uniformity tester in a black-box manner.}

\new{We show that} Theorem~\ref{thm:main} has a wide range of applications to the problem of $L_1$ identity testing for various classes of natural and well-studied structured distributions. \new{At a high level, the main message of this work is that the $\mathcal{A}_k$ distance
can be used to characterize the sample complexity of $L_1$ identity testing for broad classes of structured distributions.} The following simple proposition underlies our approach:

\begin{proposition} \label{prop:simple}
For a distribution class $C$ \new{over $[n]$} and $\eps>0$, let $k= \new{k(C, \eps)}$ be the smallest integer such that for any $f_1, f_2 \in C$ it holds that
$\|f_1-f_2\|_1 \le \|f_1-f_2\|_{{\mathcal A}_k} + \eps/2$. Then there exists an $L_1$ identity testing algorithm for $C$ using $O(\sqrt{k}/\eps^2)$ samples.
\end{proposition}

The proof of the proposition is straightforward: Given sample access to $q \in C$ and an explicit description of  $p \in C$,
we apply the $\mathcal{A}_k$-identity testing algorithm of Theorem~\ref{thm:main} for the value of $k$ in the statement of the proposition,
and error $\eps' = \eps/2$. If $q = p$, the algorithm will output ``YES'' with probability at least $2/3$. If $\|q - p\|_1 \ge \eps$, then by the condition of
Proposition~\ref{prop:simple} we have that $\|q-p\|_{\mathcal{A}_k} \ge \eps'$, and the algorithm will output ``NO'' with probability at least $2/3$.
Hence, as long as the underlying distribution satisfies the condition of Proposition~\ref{prop:simple} for a value of $k = o(n)$, Theorem~\ref{thm:main}
yields an asymptotic improvement over the sample complexity of $\Theta(\sqrt{n}/\eps^2)$.

\new{We remark that the value of $k$ in the proposition is a natural complexity measure for the difference between two probability density functions in the class $C$.
It follows from the definition of the ${\mathcal A}_k$ distance that this value corresponds to the number of ``essential'' crossings between $f_1$ and $f_2$ --  i.e.,
the number of crossings between the functions $f_1$ and $f_2$ that significantly affect their $L_1$ distance. Intuitively, the number of essential crossings -- as opposed
to the domain size -- is, in some sense, the ``right'' parameter to characterize the sample complexity of $L_1$ identity testing for $C$. As we explain below, the upper bound
implied by the above proposition is information-theoretically optimal for a wide range of structured distribution classes $C$.}


More specifically, our framework can be applied to all structured distribution classes $C$ that can be well-approximated in $L_1$ distance by {\em
piecewise low-degree polynomials}. 
We say that a distribution $p$ over $[n]$ is {\em $t$-piecewise degree-$d$}
if there exists a partition of $[n]$ into $t$ intervals such that $p$ is a (discrete) degree-$d$ polynomial within each interval.
Let ${\cal P}_{t, d}$ denote the class of all $t$-piecewise degree-$d$ distributions over $[n]$. We say that a distribution class $C$
is {\em $\eps$-close in $L_1$} to ${\cal P}_{t, d}$ if for any $f \in C$ there exists $p \in {\cal P}_{t, d}$ such that $\|f-p\|_1 \le \eps.$
It is easy to see that any pair of distributions $p, q \in {\cal P}_{t, d}$
have at most $2t(d+1)$ crossings, which implies that $\|p - q\|_{A_k} = \|p-q\|_1$, for $k = 2t(d+1)$ (see e.g., Proposition 6 in~\cite{CDSS14}).
We therefore obtain the following:
\begin{corollary} \label{cor:approx}
Let $C$ be a distribution class \new{over $[n]$} \new{and $\eps>0$}. \new{Consider parameters $t = t(C, \eps)$
and $d = d(C, \eps)$} such that $C$ is $\eps/4$-close in $L_1$ to ${\cal P}_{t, d}$.
Then there exists an $L_1$ identity testing algorithm for $C$ using $O(\sqrt{t (d+1)}/\eps^2)$ samples.
\end{corollary}
 \new{Note that any pair of values $(t, d)$ satisfying the condition above suffices for the conclusion of the corollary.
 Since our goal is to minimize the sample complexity,}
for a given class $C$, we would like to apply the corollary for values $t$ and $d$ satisfying the above condition and are
such that the product $t(d+1)$ is minimized. \new{The appropriate choice of these values is crucial,
and is based on properties of the underlying distribution family.}
Observe that the sample bound of $O(\sqrt{t (d+1)}/\eps^2)$ is tight in general,
as follows by selecting $C = {\cal P}_{t, d}$. This can be deduced from the general lower bound of $\Omega(\sqrt{n}/\eps^2)$ for uniformity testing,
and the fact that for $n = t(d+1)$, any distribution over support $[n]$ can be expressed as a $t$-piecewise degree-$d$ distribution.

The concrete testing results of Table~\ref{table:results} are obtained from Corollary~\ref{cor:approx}
by using known existential approximation theorems~\cite{Birge:87, CDSS13, CDSS14} for the corresponding structured distribution classes.
In particular, we obtain efficient identity testers, in most cases with provably optimal sample complexity,
for all the structured distribution classes studied in~\cite{CDSS13, CDSS14} in the context of learning.
Perhaps surprisingly, our upper bounds are tight not only for the class of piecewise polynomials, but also for the specific shape restricted
classes of Table~\ref{table:results}. The corresponding lower bounds for specific classes are either known from previous work (as e.g.,
in the case of $t$-modal distributions~\cite{DDSVV13}) or can be obtained using standard constructions.

\new{Finally, we remark that the results of this paper can be appropriately generalized to the setting of testing the identity of
continuous distributions over the real line. More specifically,
Theorem~\ref{thm:main} also holds for probability distributions over $\R$.
(The only additional assumption required is that the explicitly given continuous pdf $p$
can be efficiently integrated up to any additive accuracy.) In fact, the proof for the discrete setting
extends almost verbatim to the continuous setting with minor modifications.
It is easy to see that both Proposition~\ref{prop:simple} and Corollary~\ref{cor:approx} hold for the continuous setting as well.}

\subsection{Our Techniques} \label{ssec:techniques}
We now provide a detailed intuitive explanation of the ideas that lead to our main result, Theorem~\ref{thm:main}.
Given sample access to a distribution $q$ and an explicit distribution $p$, we want to test whether $q = p$ versus $\|q-p\|_{\mathcal{A}_k} \ge \eps$.
By definition we have that $\|q-p\|_{\mathcal{A}_k} = \max_{\mathcal{I}} \| q_r^{\mathcal{I} } - p_r^{\mathcal{I} } \|_1$.
So, if the ``optimal'' partition $\mathcal{J}^{\ast} = (J_i^{\ast})_{j=1}^k$ maximizing this expression
was known a priori, the problem would be easy: Our algorithm could then consider the reduced distributions
$q_r^{\mathcal{J}^{\ast}}$ and $p_r^{\mathcal{J}^{\ast}}$, which are supported on sets of size $k$, and call a standard $L_1$-identity tester
to decide whether $q_r^{\mathcal{J}^{\ast}} = p_r^{\mathcal{J}^{\ast}}$ versus $\|q_r^{\mathcal{J}^{\ast}} - p_r^{\mathcal{J}^{\ast}}\|_1 \ge \eps$.
(Note that for any given partition $\mathcal{I}$ of $[n]$ into intervals and any distribution $q$, given sample access to $q$ one
can simulate sample access to the reduced distribution $q_r^{\mathcal{I}}$.) The difficulty, of course, is that the optimal $k$-partition is {\em not
fixed}, as it depends on the unknown distribution $q$, thus it is not available to the algorithm. Hence, a more refined approach is necessary.

Our starting point is a \new{new,} simple reduction of the general problem of identity testing to its special case of uniformity testing.
\new{The main idea of the reduction is to appropriately ``stretch'' the domain size, using the explicit distribution $p$,
in order to transform the identity testing problem between $q$ and $p$
into a uniformity testing problem for a (different) distribution $q'$ (that depends on $q$ and $p$). To show correctness
of this reduction we need to show that it preserves the $\mathcal{A}_k$ distance, and that we can sample from $q'$ given samples from $q$.}

We now proceed with the details. Since $p$ is given explicitly in the input, \new{we assume for simplicity that} each $p_i$ is a rational number,
hence there exists some (potentially large) $N \in \Z_+$ such that
$p_i = \alpha_i / N$, where $\alpha_i \in \Z_{+}$ and $\sum_{i=1}^n \alpha_i = N.$\footnote{\new{We remark that this assumption is not necessary:
For the case of irrational $p_i$'s we can approximate them by rational numbers $\widetilde{p}_i$ up to sufficient accuracy and proceed
with the approximate distribution $\widetilde{p}$. This approximation step does not preserve perfect completeness; however, we point out that
our testers have some mild robustness in the completeness case, which suffices for all the arguments to go through.}}
 Given sample access to $q$ and an explicit $p$ over $[n]$, we construct an instance of the uniformity testing problem as follows:
Let $p'$ be the uniform distribution over $[N]$
and let $q'$ be the distribution over $[N]$ obtained from $q$ by subdividing the probability mass of $q_i$, $i \in [n]$, equally among
$\alpha_i$ new consecutive points. It is clear that this reduction preserves the $\mathcal{A}_{k}$ distance, i.e.,
$\|q-p\|_{\mathcal{A}_k}  = \|q'-p'\|_{\mathcal{A}_{k}}.$
The only remaining task is to show how to simulate sample access to $q'$, given samples from $q$.
Given a sample $i$ from $q$, our sample for $q'$ is selected uniformly at random from the corresponding set of $\alpha_i$ many
new points. Hence, we have reduced the problem of identity testing between $q$ and $p$ in ${\mathcal{A}_k}$ distance, to the problem
of uniformity testing of $q'$ in $\mathcal{A}_{k}$ distance. Note that this reduction is also computationally efficient,
as it only requires $O(n)$ pre-computation to specify the new intervals.


For the rest of this section, we focus on the problem of $\mathcal{A}_{k}$ uniformity testing. For notational convenience, we will use $q$ to denote
the unknown distribution and $p$ to denote the uniform distribution over $[n]$. The rough idea
is to consider an appropriate collection of interval partitions of $[n]$ and call a standard $L_1$-uniformity tester
for each of these partitions. To make such an approach work and give us a {\em sample optimal} algorithm for our
$\mathcal{A}_{k}$-uniformity testing problem we need to use a subtle and strong property of uniformity testing, \new{namely its performance guarantee
under the $L_2$ norm}.
\new{We elaborate on this point below}.

For any partition $\mathcal{I}$ of $[n]$ into $k$ intervals
by definition we have
that $\|q_r^{\mathcal{I}} - p_r^{\mathcal{I}}\|_1 \le \|q - p\|_{\mathcal{A}_k}.$ Therefore, if $q = p$, we will also have
$q_r^{\mathcal{I}} = p_r^{\mathcal{I}}$. The issue is that $\|q_r^{\mathcal{I}} - p_r^{\mathcal{I}}\|_1$ can be much smaller
than  $\|q - p\|_{\mathcal{A}_k}$; in fact, it is not difficult to construct examples where $\|q - p\|_{\mathcal{A}_k} = \Omega(1)$
and $\|q_r^{\mathcal{I}} - p_r^{\mathcal{I}}\|_1 = 0.$ In particular, it is possible for the points where $q$ is larger than $p$,
and where it is smaller than $p$ to cancel each other out within each interval in the partition,
thus making the partition useless for distinguishing $q$ from $p$.
In other words, if the partition $\mathcal{I}$ is not ``good'', we may not be able to detect any existing discrepancy.
A simple, but suboptimal, way to circumvent this issue is to consider a partition $\mathcal{I}'$ of $[n]$ into $k' =  \Theta(k/\eps)$ intervals
of the same length. Note that each such interval will have probability mass $1/k' = \Theta(\eps/k)$ under the uniform distribution $p$.
If the constant in the big-$\Theta$ is appropriately selected, say $k' = 10k/\eps$,
it is not hard to show that $\|q_r^{\mathcal{I}'} - p_r^{\mathcal{I}'}\|_1 \ge \|q - p\|_{\mathcal{A}_k} - \eps/2$; hence,
we will necessarily detect a large discrepancy for the reduced distribution. By applying the optimal $L_1$ uniformity tester
this approach will require $\Omega(\sqrt{k'} / \eps^2) = \Omega(\sqrt{k}/\eps^{2.5})$ samples.

A key tool that is essential in our analysis is a strong property of uniformity testing. An optimal $L_1$ uniformity
tester for $q$ can distinguish between the uniform distribution and the case that $\|q - p\|_1\ge \eps$ using $O(\sqrt{n}/\eps^2)$ samples.
However, a stronger guarantee is possible: With the {\em same} sample size, we can distinguish the uniform distribution
from the case that $\|q - p\|_2 \ge \eps/\sqrt{n}$. We emphasize that such a strong $L_2$ guarantee is specific to uniformity testing, and is
provably not possible for the general problem of identity testing. In previous work, Goldreich and Ron~\cite{GR00}
gave such an $L_2$ guarantee for uniformity testing, but their algorithm uses $O(\sqrt{n}/\eps^4)$ samples.
Paninski's $O(\sqrt{n}/\eps^2)$ uniformity tester works for the $L_1$ norm, and it is not known whether it achieves the desired $L_2$ property.
As one of our main tools we show \new{the following $L_2$ guarantee, which is optimal as a function of $n$ and $\eps$}:

\begin{theorem} \label{thm:unif_l2-delta}
Given $0< \eps, \delta < 1$ and sample access to a distribution $q$ over $[n]$, there is an algorithm
{\em Test-Uniformity-}$L_2 (q, n, \eps, \delta)$
which uses $m = O\left( (\sqrt{n}/\eps^2) \cdot \log(1/\delta) \right)$ samples from $q$,
runs in time linear in its sample size, and with probability at least $1-\delta$ distinguishes whether $q = U_n$ versus $\|p-q\|_{2} \ge \eps/\sqrt{n}$.
\end{theorem}

\new{To prove Theorem~\ref{thm:unif_l2-delta}} we show that a  variant of Pearson's chi-squared test~\cite{Pearson1900} -- which
can be viewed as a special case of the recent ``chi-square type'' testers in~\cite{CDVV14, VV14} -- has the desired $L_2$ guarantee.
\new{While this tester has been (implicitly) studied in~\cite{CDVV14, VV14}, and it is known to be sample optimal
with respect to the $L_1$ norm, it has not been previously analyzed for the $L_2$ norm. The novelty of Theorem~\ref{thm:unif_l2-delta}
lies in the tight analysis of the algorithm under the $L_2$ distance, and is presented in Appendix~\ref{sec:unif-L2}.
}

Armed with Theorem~\ref{thm:unif_l2-delta} we proceed as follows:
We consider a set of $j_0 = O(\log(1/\eps))$ different partitions of the domain $[n]$ into intervals. For $0 \le j < j_0$ the partition $\mathcal{I}^{(j)}$ consists of $\ell_j \eqdef  |\mathcal{I}^{(j)}| = k \cdot 2^j$ many intervals $I_i^{(j)}$, $i \in [\ell_j]$, i.e., $\mathcal{I}^{(j)} = (I_i^{(j)})_{i=1}^{\ell_j}$. For a fixed value of $j$, all intervals in $\mathcal{I}^{(j)}$ have the same length, or equivalently, the same probability mass {\em under the uniform distribution}.
Then, for any fixed $j \in [j_0]$, we have $p(I_i^{(j)}) = 1/(k \cdot 2^j)$ for all $i \in [\ell_j].$
(Observe that, by our aforementioned reduction to the uniform case,
we may assume that the domain size $n$ is a multiple of $k2^{j_0}$,
and thus that it is possible to evenly divide into such intervals of the same length).

Note that if $q = p$, then for all $0 \le j < j_0$, it holds $q_r^{\mathcal{I}^{(j)}} = p_r^{\mathcal{I}^{(j)}}$. Recalling that all intervals in $\mathcal{I}^{(j)}$
have the same probability mass {\em under $p$}, it follows that $p_r^{\mathcal{I}^{(j)}} = U_{\ell_j}$,  i.e., $p_r^{\mathcal{I}^{(j)}}$
is the uniform distribution over its support. So, if $q = p$, for any partition we have $q_r^{\mathcal{I}^{(j)}} = U_{\ell_j}$. Our main structural result (Lemma~\ref{lem:structural}) is a robust inverse lemma: {\em If $q$ is far from uniform in $\mathcal{A}_k$ distance then, for at least one of the partitions $\mathcal{I}^{(j)}$, the reduced distribution $q_r^{\mathcal{I}^{(j)}}$ will be far from uniform in $L_2$ distance.} The quantitative version of this statement is quite subtle. In particular, we start from the assumption of being $\eps$-far in $\mathcal{A}_k$ distance and can only deduce ``far'' in $L_2$ distance. This is absolutely critical for us to be able to obtain the optimal sample complexity.


\new{The key insight  for the analysis comes from noting that the optimal partition separating $q$ from $p$ in $\mathcal{A}_k$ distance cannot have too many parts. Thus, if the ``highs'' and ``lows'' cancel out over some small intervals, they must be very large in order to compensate for the fact that they are relatively narrow. Therefore, when $p$ and $q$ differ on a smaller scale, their $L_2$ discrepancy will be greater, and this compensates for the fact that the partition detecting this discrepancy will need to have more intervals in it.

}

In Section~\ref{sec:main} we present our sample optimal uniformity tester under the $\mathcal{A}_k$ distance, thereby
establishing Theorem~\ref{thm:main}. 

\section{Testing Uniformity under the ${\mathcal A}_k$-norm} \label{sec:main}



\smallskip

\fbox{\parbox{6.2in}{
{\bf Algorithm} Test-Uniformity-$\mathcal{A}_k(q, n, \eps)$\\
Input: sample access to a distribution $q$ over $[n]$, $k \in \Z_+$ with $2 \le k \le n$, and $\eps>0$.\\
Output: ``YES'' if $q = U_n$; ``NO'' if $\|q-U_n\|_{\mathcal{A}_k} \ge \eps.$

\begin{enumerate}

\item Draw a sample $S$ of size $m = O(\sqrt{k}/\eps^2)$ from $q$.


\item Fix $j_0 \in \Z_+$ such that $j_0  \eqdef \lceil \log_2(1/\eps) \rceil+O(1).$ Consider the collection
$\{\mathcal{I}^{(j)}\}_{j=0}^{j_0-1}$ of $j_0$ partitions of $[n]$ into intervals;
the partition $\mathcal{I}^{(j)} = (I_i^{(j)})_{i=1}^{\ell_j}$ consists of $\ell_j = k \cdot 2^j$ many intervals with
$p(I_i^{(j)}) = 1/(k \cdot 2^j)$, where $p \eqdef U_n$.


  \item For $j=0, 1, \ldots, j_0-1$:
	\begin{enumerate}
		\item Consider the reduced distributions $q_r^{\mathcal{I}^{(j)}}$ and $p_r^{\mathcal{I}^{(j)}} \equiv U_{\ell_j}$.
		          Use the sample $S$ to simulate samples to $q_r^{\mathcal{I}^{(j)}}$.
		\item  Run Test-Uniformity-${L_2}(q_r^{\mathcal{I}^{(j)}}, \ell_j, \eps_j, \delta_j)$
		for $\eps_j = C \cdot \eps \cdot 2^{3j/8}$ for $C>0$ a sufficiently small constant and $\delta_j =  2^{-j}/6$,
		i.e., test whether $q_r^{\mathcal{I}^{(j)}} = U_{\ell_j}$ versus $\| q_r^{\mathcal{I}^{(j)}} - U_{\ell_j}\|_2 > \gamma_j \eqdef \eps_j / \sqrt{\ell_j}.$
	\end{enumerate}
  \item If all the testers in Step~3(b) output  ``YES'', then output ``YES'';  otherwise output ``NO''.
\end{enumerate}
}}



\begin{proposition}\label{prop:ak}
The algorithm {\em Test-Uniformity-}$\mathcal{A}_k(q, n, \eps)$, on input a sample of size $m = O(\sqrt{k}/\eps^2)$ drawn from a distribution $q$ over $[n]$,
$\eps>0$ and an integer $k$ with $2 \le k \le n$, correctly distinguishes the case that $q = U_n$ from the case that $\|q-U_n\|_{\mathcal{A}_k} \ge \eps$, with probability at least $2/3$.
\end{proposition}
\begin{proof}
First, it is straightforward to verify the claimed sample complexity, as the algorithm only draws samples in Step~1.
Note that the algorithm uses the same set of samples $S$ for all testers in Step~3(b).
By Theorem~\ref{thm:unif_l2-delta}, the tester Test-Uniformity-${L_2}(q_r^{\mathcal{I}^{(j)}}, \ell_j, \eps_j, \delta_j)$,
on input a set of $m_j = O((\sqrt{\ell_j}/\eps_j^2) \cdot \log(1/\delta_j))$ samples from $q_r^{\mathcal{I}^{(j)}}$ distinguishes the case
that $q_r^{\mathcal{I}^{(j)}} = U_{\ell_j}$ from the case that  $\|q_r^{\mathcal{I}^{(j)}} - U_{\ell_j}\|_2 \ge \gamma_j \eqdef \eps_j / \sqrt{\ell_j}$ with probability at least
$1-\delta_j$. From our choice of parameters it can be verified that $\max_j m_j \le m= O(\sqrt{k}/\eps^2)$, hence we can use the same sample $S$ as input to
these testers for all $0 \le j \le j_0-1$. In fact, it is easy to see that $\sum_{j=0}^{j_0-1} m_j = O(m)$, which implies that the overall algorithm runs in sample-linear time.
Since each tester in Step 3(b) has error probability $\delta_j$, by a union bound over all $j \in \{0, \ldots, j_0-1\}$, the total error probability is at most
$\sum_{j=0}^{j_0-1} \delta_j  \le (1/6) \cdot \sum_{j=0}^{\infty} 2^{-j} = 1/3.$
Therefore, with probability at least $2/3$ all the testers in Step~3(b) succeed.
We will henceforth condition on this ``good'' event, and establish the completeness and soundness properties of
the overall algorithm under this conditioning.

We start by establishing completeness. If $q = p = U_n$, then for any partition $\mathcal{I}^{(j)}$, $0 \le j \le j_0-1$, we have that
$q_r^{\mathcal{I}^{(j)}} = p_r^{\mathcal{I}^{(j)}} = U_{\ell_j}$. By our aforementioned conditioning, all testers in Step~3(b) will output ``YES'',
hence the overall algorithm will also output ``YES'', as desired.

We now proceed to establish the soundness of our algorithm.
Assuming that $\|q - p\|_{\mathcal{A}_k} \ge \eps$, we want to show that the algorithm
Test-Uniformity-$\mathcal{A}_k(q, n, \eps)$ outputs ``NO'' with probability at least $2/3$.
Towards this end, we prove the following structural lemma:
\begin{lemma} \label{lem:structural} \new{There exists a constant $C>0$ such that the following holds:}
If $\|q - p\|_{\mathcal{A}_k} \ge \eps$, there exists $j \in \Z_+$ with $0 \le j \le j_0-1$ such that
$\| q_r^{\mathcal{I}^{(j)}} - U_{\ell_j}\|^2_2 \ge \gamma_j^2 \new{\eqdef \eps^2_j /\ell_j = C^2 \cdot (\eps^2/k) \cdot 2^{-j/4}.}$
\end{lemma}
Given the lemma, the soundness property of our algorithm follows easily.
Indeed, since all testers Test-Uniformity-${L_2}(q_r^{\mathcal{I}^{(j)}}, \ell_j, \eps_j, \delta_j)$ of Step~3(b) are successful by our conditioning,
Lemma~~\ref{lem:structural} implies that at least one of them outputs ``NO'', hence the overall algorithm will output ``NO''.
\end{proof}

\noindent The proof of Lemma~\ref{lem:structural} in its full generality is quite technical.
For the sake of the intuition, in the following subsection (Section~\ref{ssec:kflat-struct})
we provide a proof of the lemma for the important special case that the unknown distribution $q$ is promised to be {\em $k$-flat}, i.e.,
piecewise constant with $k$ pieces. This setting captures many of the core ideas and, at the same time, avoids some of the necessary technical difficulties
of the general case. Finally, in Section~\ref{sec:ak} we present our proof for the general case.

\subsection{Proof of Structural Lemma: $k$-flat Case} \label{ssec:kflat-struct}

\new{For this special case we will prove the lemma for $C = 1/80$.}
Since $q$ is $k$-flat there exists a partition $\mathcal{I}^{\ast} = (I_j^{\ast})_{j=1}^k$
of $[n]$ into $k$ intervals so that $q$ is constant within each such interval.
This in particular implies that $\|q - p\|_{\mathcal{A}_k}  = \|q - p\|_1$, where $p \eqdef U_n$.
For $J \in \mathcal{I}^{\ast}$
let us denote by $q_{J}$ the value of $q$ within interval
$J$, that is, for all $j \in [k]$ and $i \in I_j^{\ast}$ we have $q_i = q_{I_j^{\ast}}$. For notational convenience, we sometimes use
$p_{J}$ to denote the value of $p = U_n$ within interval $J$. 
By assumption we have that
$\|q-p\|_1 = 
\sum_{j=1}^k  |I_j^{\ast}| \cdot |q_{I_j^{\ast}} - 1/n|  \ge \eps.$

Throughout the proof, we work with intervals $I_j^{\ast} \in \mathcal{I}^{\ast}$ such that $q_{I_j^{\ast}} < 1/n.$
We will henceforth refer to such intervals as {\em troughs} and will denote by $\mathcal{T} \subseteq [k]$ the corresponding set of indices, i.e.,
$\mathcal{T} = \{ j \in [k] \mid q_{I_j^{\ast}} < 1/n \}$.
For each trough $J \in \{I_j^{\ast}\}_{j \in \mathcal{T}}$ we define its {\em depth} as
$\depth(J) = (p_J - q_J)/p_J =  n \cdot (1/n - q_J)$
and its {\em width} as
$\width(J) = p(J) = (1/n) \cdot |J|.$
Note that the width of $J$ is identified with the probability mass that the uniform distribution assigns to it.
The {\em discrepancy} of a trough $J$ is defined by
$\discr(J) = \depth(J) \cdot   \width(J) = |J| \cdot \left(1/n - q_J \right)$
and corresponds to the contribution of $J$ to the $L_1$ distance between $q$ and $p$.

It follows from Scheffe's identity that half of the contribution to $\|q-p\|_1$ comes
from troughs, namely
$\|q-p\|_1^{\mathcal{T}} \eqdef  \sum_{j \in \mathcal{T}}   \discr(I_j^{\ast}) = (1/2) \cdot \|q-p\|_1 \ge \eps/2.$
\new{An important observation is that we may assume that all troughs have {\em width} at most $1/k$ at the cost of potentially doubling the total number of intervals}.
Indeed, it is easy to see that we can artificially subdivide ``wider'' troughs so that each new trough has width at most $1/k$. This process
comes at the expense of at most doubling the number of troughs. Let us denote by $\{\widetilde{I_j}\}_{j \in \mathcal{T}'}$ this set of (new) troughs,
where $|\mathcal{T}'| \le 2k$ and each $\widetilde{I_j}$ is a subset of some $I_i^{\ast}$, $i \in \mathcal{T}.$
We will henceforth deal with the set of troughs $\{\widetilde{I_j}\}_{j \in \mathcal{T}'}$ each of width at most $1/k$.
By construction, it is clear that
\begin{equation} \label{eqn:ena-b}
\|q-p\|_1^{\mathcal{T}'} \eqdef  \sum_{j \in \mathcal{T}'} \discr(\widetilde{I_j}) =
\|q-p\|_1^{\mathcal{T}} \ge \eps/2.
\end{equation}
At this point we note that we can essentially ignore troughs $J \in \{\widetilde{I_j}\}_{j \in \mathcal{T}'}$ with small discrepancy.
Indeed,  the total contribution of  intervals $J \in \{\widetilde{I_j}\}_{j \in \mathcal{T}'}$
with $\discr(J) \le \eps/20k$ to the LHS of (\ref{eqn:ena-b}) is at most $|\mathcal{T}'| \cdot \new{(}\eps/20k) \le  2k \cdot (\eps/20k) = \eps/10$.
Let $\mathcal{T}^{\ast}$ be the subset of $\mathcal{T}'$ corresponding to troughs with discrepancy at least $\eps/20k$, i.e.,
$j \in \mathcal{T}^{\ast}$ if and only if $j \in \mathcal{T}'$ and $ \discr(\widetilde{I_j}) \ge \eps/20k.$ Then, we have that
\begin{equation} \label{eqn:dyo}
\|q-p\|_1^{\mathcal{T}^{\ast}} \eqdef  \sum_{j \in \mathcal{T}^{\ast}} \discr(\widetilde{I_j}) \ge 2\eps/5.
\end{equation}
Observe that for any interval $J$ it holds $\discr(J) \le \width(J)$. \new{Note that this part of the argument depends critically on considering only troughs.} Hence, for $j \in \mathcal{T}^{\ast}$
we have that
\begin{equation} \label{eqn:width}
\eps/(20k) \le \width(\widetilde{I_j}) \le 1/k.
\end{equation}
Thus far we have argued that a constant fraction of the contribution to $\|q-p\|_1$ comes
from troughs whose width satisfies (\ref{eqn:width}).  Our next crucial claim is that each such trough
must have a ``large'' overlap with one of the intervals $I_i^{(j)}$ considered by our algorithm Test-Uniformity-$\mathcal{A}_k$.
In particular, consider a trough $J \in \{\widetilde{I_j}\}_{j \in \mathcal{T}^{\ast}}$.
We claim that there exists $j \in \{ 0, \ldots, j_0-1\}$ and $i \in [\ell_{j}]$
such that $| I_{i}^{(j)}| \ge |J|/\new{4}$ \new{and so that $I_i^{(j)}\subseteq J$.
To see this we first pick a $j$ so that $\width(J)/2 > 2^{-j}/k \geq \width(J)/4.$ Since the $I^{(j)}_i$ have width less than half that of $J$, $J$ must intersect at least three of these intervals. Thus, any but the two outermost such intervals will be entirely contained within $J$, and furthermore has width $2^j/k \geq \width(J)/4.$
}


Since the interval $L \in \mathcal{I}^{(j+1)}$ is a ``domain point'' for the reduced distribution $q_{r}^{\mathcal{I}^{(j+1)}}$, the $L_1$ error between
$q_{r}^{\mathcal{I}^{(j+1)}}$ and $U_{\ell_{j+1}}$ incurred by this element is at least $\frac{1}{\new{4}}\cdot \discr(J)$, and
the corresponding $L_2^2$ error is at least $\frac{1}{\new{16}} \cdot (\discr(J))^2 \ge \frac{\eps}{\new{320}k} \cdot \discr(J),$ where the inequality follows from the fact
that $\discr(J) \ge \eps/(20k)$.
Hence, we have that
\begin{equation} \label{eqn:tria}
\| q_{r}^{\mathcal{I}^{(j+1)}} - U_{\ell_{j+1}} \|_2^2 \ge  \eps/(\new{320}k) \cdot \discr(J).
\end{equation}
As shown above, for every trough $J \in \{\widetilde{I_j}\}_{j \in \mathcal{T}^{\ast}}$ there exists a level $j \in \{0, \ldots, j_0-1\}$
such that (\ref{eqn:tria}) holds.  Hence, summing (\ref{eqn:tria}) over all levels we obtain
\begin{equation} \label{eqn:tessa}
\sum_{j=0}^{j_0-1} \| q_{r}^{\mathcal{I}^{(j+1)}} - U_{\ell_{j+1}} \|_2^2
\ge  \eps/(\new{320}k) \cdot   \sum_{j \in \mathcal{T}^{\ast}} \discr(\widetilde{I_j})
\ge   \eps^2 / (\new{800}k),
\end{equation}
where the second inequality follows from (\ref{eqn:dyo}).
Note that
$$
\sum_{j=0}^{j_0-1} \gamma_j^2 \leq \sum_{j=0}^{j_0-1} \frac{\eps^2 \cdot 2^{3j/4}}{\new{80}^2 \cdot k 2^j} = \frac{\eps^2}{\new{6400}k}\sum_{j=0}^{j_0-1} 2^{-j/4} <  \eps^2 / (\new{800}k).
$$
Therefore, by the above, we must have that
$
\| q_{r}^{\mathcal{I}^{(j+1)}} - U_{\ell_{j+1}} \|_2^2 > \gamma_j^2
$
for some $0\leq j \leq j_0-1$.
This completes the proof of Lemma~\ref{lem:structural} for the special case of $q$ being $k$-flat.

\subsection{Proof of Structural Lemma: General Case} \label{sec:ak}

\new{To prove the general version of our structural result for the $\mathcal{A}_k$ distance, we will need to choose an appropriate value for the universal constant $C$.
We show that it is sufficient to take $C\leq 5 \cdot 10^{-6}$.}
(While we have not attempted to optimize constant factors, we believe that a more careful analysis will lead to substantially better constants.)

A useful observation is that our Test-Uniformity-$\mathcal{A}_k$ algorithm only distinguishes which of the intervals of $\mathcal{I}^{(j_0-1)}$ each of our samples lies in,
and can therefore equivalently be thought of as a uniformity tester for the reduced distribution $q_r^{\mathcal{I}^{(j_0-1)}}.$ In order to show that it suffices to consider only this restricted sample set, we claim that
$$
\|q_r^{\mathcal{I}^{(j_0-1)}} - U_{\ell_{j_0-1}}\|_{\mathcal{A}_k} \geq \|p-q\|_{\mathcal{A}_k} - \eps/2.
$$
In particular, these $\mathcal{A}_k$ distances would be equal if the dividers of the optimal partition for $q$ were all on boundaries between intervals of $\mathcal{I}^{(j_0-1)}$. If this was not the case though, we could round the endpoint of each trough inward to the nearest such boundary (note that we can assume that the optimal partition has no two adjacent troughs). This increases the discrepancy of each trough by at most $2k \cdot 2^{-j_0}$, and thus for $j_0-\log_2(1/\eps)$ a sufficiently large universal constant, the total discrepancy decreases by at most $\eps/2$.

Thus, we have reduced ourselves to the case where $n=2^{j_0-1}\cdot k$ and have argued that it suffices to show that our algorithm works to distinguish $\mathcal{A}_k$-distance in this setting with $\eps_j= 10^{-5} \cdot \eps \cdot 2^{3j/8}$.

The analysis of the completeness and the soundness of the tester is identical to Proposition~\ref{prop:ak}.
The only missing piece is the proof of Lemma~\ref{lem:structural}, which we now restate for the sake of convenience:

\begin{lemma} \label{lem:structuralAk}
If $\|q - p\|_{\mathcal{A}_k} \ge \eps$, there exists some $j \in \Z_+$ with $0 \le j \le j_0-1$ such that
$$\| q_r^{\mathcal{I}^{(j)}} - U_{\ell_j}\|^2_2 \ge {\gamma_j^2} := \eps_j^2/\ell_j = 10^{-10} 2^{-j/4} \eps^2/k.$$
\end{lemma}

The analysis of the general case here is somewhat more complicated than the special case for $q$ being $k$-flat case that was presented in the previous section.
This is because it is possible for one of the intervals $J$ in the optimal partition (i.e., the interval partition $\mathcal{I}^{\ast} \in \mathfrak{J}_k$ maximizing
$\|q_r^{\mathcal{I}} - q_r^{\mathcal{I}}\|_1$ in the definition of the $\mathcal{A}_k$ distance) to have large overlap with an interval $I$ that our algorithm considers -- that is,
$I \in \cup_{j = 0}^{j_0-1} \mathcal{I}^{(j)}$ --
without having $q(I)$ and $p(I)$ differ substantially. Note that the unknown distribution $q$ is not guaranteed to be constant within such an interval $J$,
and in particular the difference $q-p$ does not necessarily preserve its sign within $J$.

To deal with this issue, we note that there are two possibilities for an interval $J$ in the optimal partition: Either one of the intervals $I_i^{(j)}$ (considered by our algorithm) of size at least $|J|/2$ has discrepancy comparable to $J$, or the distribution $q$ differs from $p$ even more substantially on one of the intervals separating the endpointss
of $I_i^{(j)}$ from the endpoints of $J$. Therefore, either an interval contained within this will detect a large $L_2$ error, or we will need to again pass to a subinterval.
To make this intuition rigorous, we will need a mechanism for detecting where this recursion will terminate. To handle this formally, we introduce the following definition:
\begin{definition}
For $p=U_n$ and $q$ an arbitrary distribution over $[n]$, we
define the {\em scale-sensitive-$L_2$ distance} between $q$ and $p$ to be
$$
\|q - p\|^2_{[k]} \eqdef \max_{\mathcal{I} = (I_1, \ldots, I_r) \in \mathbf{W}_{1/k}} \sum_{i=1}^r \frac{\discr^2(I_i)}{\width^{1/8}(I_i)}
$$
where $ \mathbf{W}_{1/k}$ is the collection of all interval partitions of $[n]$ into intervals of width at most $1/k$.
\end{definition}



The notion of the scale-sensitive-$L_2$ distance will be a useful intermediate tool in our analysis.
The rough idea of the definition is that the optimal partition will be able to detect the correctly sized intervals for our tester to notice.
(It will act as an analogue of the partition into the intervals where $q$ is constant for the $k$-flat case.)

The first thing we need to show is that if $q$ and $p$ have large $\mathcal{A}_k$ distance then they also have large scale-sensitive-$L_2$ distance.
Indeed, we have the following lemma:
\begin{lemma}\label{AkScaledL2Lem}
For $p=U_n$ and $q$ an arbitrary distribution over $[n]$, we have that
$$
\|q-p\|^2_{[k]} \geq \frac{\|q-p\|_{\mathcal{A}_k}^2}{(2k)^{7/8}}.
$$
\end{lemma}
\begin{proof}
Let $\eps=\|q - p\|_{\mathcal{A}_k}^2$. Consider the optimal $\mathcal{I}^{\ast}$ in the definition of the $\mathcal{A}_k$ distance. As in our analysis for the $k$-flat case,
by further subdividing intervals of width more than $1/k$ into smaller ones, we can obtain a new partition, $\mathcal{I}' = (I_i')_{i=1}^{s}$, of cardinality $s \le 2k$ all of whose parts have width at most $1/k$. Furthermore, we have that $\sum_i \discr(I'_i) \geq \eps$. Using this partition to bound from below $\|q-p\|^2_{[k]}$, by Cauchy-Schwarz we obtain that
\begin{align*}
\|q-p\|^2_{[k]} & \geq \sum_i \frac{\discr^2(I'_i)}{\width(I'_i)^{1/8}}\\
& \geq \frac{\left(\sum_i \discr(I'_i) \right)^2}{\sum_i \width(I'_i)^{1/8}}\\
& \geq \frac{\eps^2}{2k (1/(2k))^{1/8}} \\ & = \frac{\eps^2}{(2k)^{7/8}}.
\end{align*}
\end{proof}

The second important fact about the scale-sensitive-$L_2$ distance is that if it is large then one of the partitions considered in our algorithm will produce a large $L_2$ error.
\begin{proposition}\label{scaledL2IntProp}
Let $p=U_n$ be the uniform distribution and $q$ a distribution over $[n]$. Then we have that
\begin{equation}\label{scaledDistErrsEqn}
\|q-p\|_{[k]}^2 \leq 10^8 \sum_{j=0}^{j_0-1} \sum_{i=1}^{2^j \cdot k} \frac{\discr^2(I_i^{(j)})}{\width^{1/8}(I_i^{(j)})}.
\end{equation}
\end{proposition}
\begin{proof}
Let $\mathcal{J}  \in \mathbf{W}_{1/k}$ be the optimal partition used when computing the scale-sensitive-$L_2$ distance $\|q-p\|_{[k]}$. In particular, it is a partition into intervals of width at most $1/k$ so that $\sum_i \frac{\discr^2(J_i)}{\width(J_i)^{1/8}}$ is maximized. To prove Equation \eqref{scaledDistErrsEqn}, we prove a notably stronger claim. In particular, we will prove that for each interval $J_\ell\in\mathcal{J}$
\begin{equation}\label{refinedSumEqn}
\frac{\discr^2(J_\ell)}{\width^{1/8}(J_\ell)} \leq 10^8\sum_{j=0}^{j_0-1} \sum_{i: I_i^{(j)}\subset J_\ell} \frac{\discr^2(I_i^{(j)})}{\width^{1/8}(I_i^{(j)})}.
\end{equation}
Summing over $\ell$ would then yield $\|q-p\|_{[k]}^2$ on the left hand side and a strict subset of the terms from Equation \eqref{scaledDistErrsEqn} on the right hand side. From here on, we will consider only a single interval $J_\ell$. For notational convenience, we will drop the subscript and merely call it $J$.

First, note that if $|J|\leq 10^8$, then this follows easily from considering just the sum over $j=j_0-1$. Then, if $t=|J|$, $J$ is divided into $t$ intervals of size one. The sum of the discrepancies of these intervals equals the discrepancy of $J$, and thus, the sum of the squares of the discrepancies is at least $\discr^2(J)/t$. Furthermore, the widths of these subintervals are all smaller than the width of $J$ by a factor of $t$. Thus, in this case the sum of the right hand side of Equation \eqref{refinedSumEqn} is at least $1/t^{7/8}\geq \frac{1}{10^7}$ of the left hand side.

Otherwise, if $|J|>10^8$, we can find a $j$ so that $\width(J)/10^8 < 1/(2^j\cdot k) \leq 2\cdot \width(J)/10^8$. We claim that in this case Equation \eqref{refinedSumEqn} holds even if we restrict the sum on the right hand side to this value of $j$. Note that $J$ contains at most $10^8$ intervals of $\mathcal{I}^{(j)}$, and that it is covered by these intervals plus two narrower intervals on the ends. Call these end-intervals $R_1$ and $R_2$. We claim that $\discr(R_i)\leq \discr(J)/3$. This is because otherwise it would be the case that
$$
\frac{\discr^2(R_i)}{\width^{1/8}(R_i)} > \frac{\discr^2(J)}{\width^{1/8}(J)}.
$$
(This is because $(1/3)^2\cdot (2/10^8)^{-1/8} > 1$.)
This is a contradiction, since it would mean that partitioning $J$ into $R_i$ and its complement would improve the sum defining $\|q-p\|_{[k]}$, which was assumed to be maximum. This in turn implies that the sum of the discrepancies of the $I_i^{(j)}$ contained in $J$ must be at least $\discr(J)/3$, so the sum of their squares is at least $\discr^2(J)/(9\cdot 10^8)$. On the other hand, each of these intervals is narrower than $J$ by a factor of at least $10^8/2$, thus the appropriate sum of $\frac{\discr^2(I_i^{(j)})}{\width^{1/8}(I_i^{(j)})}$ is at least $\frac{\discr^2(J)}{10^8\width^{1/8}(J)}$. This completes the proof.
\end{proof}

We are now ready to prove Lemma \ref{lem:structuralAk}.
\begin{proof}
If $\|q-p\|_{\mathcal{A}_k}\geq \eps$ we have by Lemma \ref{AkScaledL2Lem} that
$$\|q-p\|_{[k]}^2 \geq \frac{\eps^2}{(2k)^{7/8}}.$$ By Proposition \ref{scaledL2IntProp}, this implies that
\begin{align*}
\frac{\eps^2}{(2k)^{7/8}} &\leq 10^8\sum_{j=0}^{j_0-1} \sum_{i=1}^{2^j\cdot k} \frac{\discr^2(I_i^{(j)})}{\width^{1/8}(I_i^{(j)})}\\
& = 10^8 \sum_{j=0}^{j_0-1} (2^{j}k)^{1/8} \|q^{\mathcal{I}^{(j)}}-U_{\ell_j}\|_2^2.
\end{align*}
Therefore,
\begin{equation}\label{L2SumEqn}
\sum_{j=0}^{j_0-1} 2^{j/8} \|q^{\mathcal{I}^{(j)}}-U_{\ell_j}\|_2^2 \geq 5\cdot 10^{-9} \eps^2/k.
\end{equation}
On the other hand, if $\|q^{\mathcal{I}^{(j)}}-U_{\ell_j}\|_2^2$ were at most $10^{-10}2^{-j/4}\eps^2/k$ for each $j$, then the sum above would be at most
$$
10^{-10}\eps^2/k \sum_j 2^{-j/8} < 5\cdot 10^{-9} \eps^2/k.
$$
This would contradict Equation \eqref{L2SumEqn}, thus proving that $\|q^{\mathcal{I}^{(j)}}-U_{\ell_j}\|_2^2\geq10^{-10}2^{-j/4}\eps^2/k$ for at least one $j$,
proving Lemma~ \ref{lem:structuralAk}.
\end{proof}

\section{Conclusions and Future Work} \label{sec:concl}

In this work we designed a computationally efficient algorithm
for the problem of identity testing against a known distribution, which yields sample optimal bounds
for a wide range of natural and important classes of structured distributions.
A natural direction for future work is to generalize our results to the problem of identity testing between two unknown structured distributions.
What is the optimal sample complexity in this more general setting?
We emphasize that new ideas are required for this problem,
as the algorithm and analysis in this work crucially exploit the a priori knowledge of the explicit distribution.

\bibliographystyle{alpha}


\bibliography{allrefs}

\newcommand{\etalchar}[1]{$^{#1}$}
\begin{thebibliography}{WWW{\etalchar{+}}05}

\bibitem[ADJ{\etalchar{+}}11]{DJOP11}
J.~Acharya, H.~Das, A.~Jafarpour, A.~Orlitsky, and S.~Pan.
\newblock Competitive closeness testing.
\newblock {\em Journal of Machine Learning Research - Proceedings Track},
  19:47--68, 2011.

\bibitem[Bat01]{Batu01}
T.~Batu.
\newblock {\em Testing Properties of Distributions}.
\newblock PhD thesis, Cornell University, 2001.

\bibitem[BBBB72]{BBBB:72}
R.E. Barlow, D.J. Bartholomew, J.M. Bremner, and H.D. Brunk.
\newblock {\em Statistical Inference under Order Restrictions}.
\newblock Wiley, New York, 1972.

\bibitem[BDKR02]{BDKR:02}
T.~Batu, S.~Dasgupta, R.~Kumar, and R.~Rubinfeld.
\newblock The complexity of approximating entropy.
\newblock In {\em {ACM} Symposium on Theory of Computing}, pages 678--687,
  2002.

\bibitem[BFF{\etalchar{+}}01]{BFFKRW:01}
T.~Batu, E.~Fischer, L.~Fortnow, R.~Kumar, R.~Rubinfeld, and P.~White.
\newblock Testing random variables for independence and identity.
\newblock In {\em Proc. 42nd IEEE Symposium on Foundations of Computer
  Science}, pages 442--451, 2001.

\bibitem[BFR{\etalchar{+}}00]{BFR+:00}
T.~Batu, L.~Fortnow, R.~Rubinfeld, W.~D. Smith, and P.~White.
\newblock Testing that distributions are close.
\newblock In {\em {IEEE} Symposium on Foundations of Computer Science}, pages
  259--269, 2000.

\bibitem[BFR{\etalchar{+}}13]{Batu13}
T.~Batu, L.~Fortnow, R.~Rubinfeld, W.~D. Smith, and P.~White.
\newblock Testing closeness of discrete distributions.
\newblock {\em J. ACM}, 60(1):4, 2013.

\bibitem[Bir87a]{Birge:87}
L.~Birg\'e.
\newblock {Estimating a density under order restrictions: Nonasymptotic minimax
  risk}.
\newblock {\em Annals of Statistics}, 15(3):995--1012, 1987.

\bibitem[Bir87b]{Birge:87b}
L.~Birg\'e.
\newblock {On the risk of histograms for estimating decreasing densities}.
\newblock {\em Annals of Statistics}, 15(3):1013--1022, 1987.

\bibitem[BKR04]{BKR:04}
T.~Batu, R.~Kumar, and R.~Rubinfeld.
\newblock Sublinear algorithms for testing monotone and unimodal distributions.
\newblock In {\em {ACM} Symposium on Theory of Computing}, pages 381--390,
  2004.

\bibitem[Bru58]{Brunk:58}
H.~D. Brunk.
\newblock On the estimation of parameters restricted by inequalities.
\newblock {\em The Annals of Mathematical Statistics}, 29(2):pp. 437--454,
  1958.

\bibitem[BRW09]{BRW:09aos}
F.~Balabdaoui, K.~Rufibach, and J.~A. Wellner.
\newblock Limit distribution theory for maximum likelihood estimation of a
  log-concave density.
\newblock {\em The Annals of Statistics}, 37(3):pp. 1299--1331, 2009.

\bibitem[BS10]{BelkinSinha:10}
M.~Belkin and K.~Sinha.
\newblock Polynomial learning of distribution families.
\newblock In {\em FOCS}, pages 103--112, 2010.

\bibitem[BW07]{BW07aos}
F.~Balabdaoui and J.~A. Wellner.
\newblock Estimation of a $k$-monotone density: Limit distribution theory and
  the spline connection.
\newblock {\em The Annals of Statistics}, 35(6):pp. 2536--2564, 2007.

\bibitem[BW10]{BW10sn}
F.~Balabdaoui and J.~A. Wellner.
\newblock Estimation of a $k$-monotone density: characterizations, consistency
  and minimax lower bounds.
\newblock {\em Statistica Neerlandica}, 64(1):45--70, 2010.

\bibitem[CDSS13]{CDSS13}
S.~Chan, I.~Diakonikolas, R.~Servedio, and X.~Sun.
\newblock Learning mixtures of structured distributions over discrete domains.
\newblock In {\em SODA}, pages 1380--1394, 2013.

\bibitem[CDSS14]{CDSS14}
S.~Chan, I.~Diakonikolas, R.~Servedio, and X.~Sun.
\newblock Efficient density estimation via piecewise polynomial approximation.
\newblock In {\em STOC}, pages 604--613, 2014.

\bibitem[CDVV14]{CDVV14}
S.~Chan, I.~Diakonikolas, P.~Valiant, and G.~Valiant.
\newblock Optimal algorithms for testing closeness of discrete distributions.
\newblock In {\em SODA}, pages 1193--1203, 2014.

\bibitem[CT04]{ChanTong:04}
K.S. Chan and H.~Tong.
\newblock Testing for multimodality with dependent data.
\newblock {\em Biometrika}, 91(1):113--123, 2004.

\bibitem[DDO{\etalchar{+}}13]{DDOST13focs}
C.~Daskalakis, I.~Diakonikolas, R.~O'Donnell, R.A. Servedio, and L.~Tan.
\newblock {Learning Sums of Independent Integer Random Variables}.
\newblock In {\em FOCS}, pages 217--226, 2013.

\bibitem[DDS12a]{DDS12soda}
C.~Daskalakis, I.~Diakonikolas, and R.A. Servedio.
\newblock Learning $k$-modal distributions via testing.
\newblock In {\em SODA}, pages 1371--1385, 2012.

\bibitem[DDS12b]{DDS12stoc}
C.~Daskalakis, I.~Diakonikolas, and R.A. Servedio.
\newblock {Learning Poisson Binomial Distributions}.
\newblock In {\em STOC}, pages 709--728, 2012.

\bibitem[DDS{\etalchar{+}}13]{DDSVV13}
C.~Daskalakis, I.~Diakonikolas, R.~Servedio, G.~Valiant, and P.~Valiant.
\newblock Testing $k$-modal distributions: Optimal algorithms via reductions.
\newblock In {\em SODA}, pages 1833--1852, 2013.

\bibitem[DL01]{DL:01}
L.~Devroye and G.~Lugosi.
\newblock {\em Combinatorial methods in density estimation}.
\newblock Springer Series in Statistics, Springer, 2001.

\bibitem[DR09]{DumbgenRufibach:09}
L.~D\:{u}mbgen and K.~Rufibach.
\newblock Maximum likelihood estimation of a log-concave density and its
  distribution function: Basic properties and uniform consistency.
\newblock {\em Bernoulli}, 15(1):40--68, 2009.

\bibitem[FOS05]{FOS:05focs}
J.~Feldman, R.~O'Donnell, and R.~Servedio.
\newblock Learning mixtures of product distributions over discrete domains.
\newblock In {\em Proc.\ 46th Symposium on Foundations of Computer Science
  (FOCS)}, pages 501--510, 2005.

\bibitem[Fou97]{Fougeres:97}
A.-L. Foug\`{e}res.
\newblock Estimation de densit\'{e}s unimodales.
\newblock {\em Canadian Journal of Statistics}, 25:375--387, 1997.

\bibitem[GGR98]{GGR98}
O.~Goldreich, S.~Goldwasser, and D.~Ron.
\newblock Property testing and its connection to learning and approximation.
\newblock {\em Journal of the ACM}, 45:653--750, 1998.

\bibitem[GR00]{GR00}
O.~Goldreich and D.~Ron.
\newblock On testing expansion in bounded-degree graphs.
\newblock Technical Report TR00-020, Electronic Colloquium on Computational
  Complexity, 2000.

\bibitem[Gre56]{Grenander:56}
U.~Grenander.
\newblock On the theory of mortality measurement.
\newblock {\em Skand. Aktuarietidskr.}, 39:125--153, 1956.

\bibitem[Gro85]{Groeneboom:85}
P.~Groeneboom.
\newblock Estimating a monotone density.
\newblock In {\em Proc. of the Berkeley Conference in Honor of Jerzy Neyman and
  Jack Kiefer}, pages 539--555, 1985.

\bibitem[GW09]{GW09sc}
F.~Gao and J.~A. Wellner.
\newblock On the rate of convergence of the maximum likelihood estimator of a
  $k$-monotone density.
\newblock {\em Science in China Series A: Mathematics}, 52:1525--1538, 2009.

\bibitem[Ham87]{Hampel87}
F.~R. Hampel.
\newblock Design, data \& analysis.
\newblock chapter Design, modelling, and analysis of some biological data sets,
  pages 93--128. John Wiley \& Sons, Inc., New York, NY, USA, 1987.

\bibitem[HP76]{HansonP:76}
D.~L. Hanson and G.~Pledger.
\newblock Consistency in concave regression.
\newblock {\em The Annals of Statistics}, 4(6):pp. 1038--1050, 1976.

\bibitem[ILR12]{ILR12}
P.~Indyk, R.~Levi, and R.~Rubinfeld.
\newblock {Approximating and Testing $k$-Histogram Distributions in Sub-linear
  Time}.
\newblock In {\em PODS}, pages 15--22, 2012.

\bibitem[JW09]{JW:09}
H.~K. Jankowski and J.~A. Wellner.
\newblock Estimation of a discrete monotone density.
\newblock {\em Electronic Journal of Statistics}, 3:1567--1605, 2009.

\bibitem[KM10]{KoenkerM:10aos}
R.~Koenker and I.~Mizera.
\newblock Quasi-concave density estimation.
\newblock {\em Ann. Statist.}, 38(5):2998--3027, 2010.

\bibitem[KMR{\etalchar{+}}94]{KMR+:94}
M.~Kearns, Y.~Mansour, D.~Ron, R.~Rubinfeld, R.~Schapire, and L.~Sellie.
\newblock On the learnability of discrete distributions.
\newblock In {\em Proceedings of the 26th Symposium on Theory of Computing},
  pages 273--282, 1994.

\bibitem[KMV10]{KMV:10}
A.~T. Kalai, A.~Moitra, and G.~Valiant.
\newblock {Efficiently learning mixtures of two Gaussians}.
\newblock In {\em STOC}, pages 553--562, 2010.

\bibitem[LRR11]{LRR11}
R.~Levi, D.~Ron, and R.~Rubinfeld.
\newblock Testing properties of collections of distributions.
\newblock In {\em ICS}, pages 179--194, 2011.

\bibitem[MV10]{MoitraValiant:10}
A.~Moitra and G.~Valiant.
\newblock {Settling the polynomial learnability of mixtures of Gaussians}.
\newblock In {\em FOCS}, pages 93--102, 2010.

\bibitem[NP33]{NeymanP}
J.~Neyman and E.~S. Pearson.
\newblock On the problem of the most efficient tests of statistical hypotheses.
\newblock {\em Philosophical Transactions of the Royal Society of London.
  Series A, Containing Papers of a Mathematical or Physical Character},
  231(694-706):289--337, 1933.

\bibitem[Pan08]{Paninski:08}
L.~Paninski.
\newblock A coincidence-based test for uniformity given very sparsely-sampled
  discrete data.
\newblock {\em IEEE Transactions on Information Theory}, 54:4750--4755, 2008.

\bibitem[Pea00]{Pearson1900}
K.~Pearson.
\newblock On the criterion that a given system of deviations from the probable
  in the case of a correlated system of variables is such that it can be
  reasonably supposed to have arisen from random sampling.
\newblock {\em Philosophical Magazine Series 5}, 50(302):157--175, 1900.

\bibitem[Rao69]{PrakasaRao:69}
B.L.S.~Prakasa Rao.
\newblock Estimation of a unimodal density.
\newblock {\em Sankhya Ser. A}, 31:23--36, 1969.

\bibitem[Reb05]{Reb05aos}
L.~Reboul.
\newblock Estimation of a function under shape restrictions. {A}pplications to
  reliability.
\newblock {\em Ann. Statist.}, 33(3):1330--1356, 2005.

\bibitem[RS96]{RS96}
R.~Rubinfeld and M.~Sudan.
\newblock Robust characterizations of polynomials with applications to program
  testing.
\newblock {\em SIAM J. on Comput.}, 25:252--271, 1996.

\bibitem[Rub12]{Rub12}
R.~Rubinfeld.
\newblock Taming big probability distributions.
\newblock {\em XRDS}, 19(1):24--28, 2012.

\bibitem[Val11]{PV11sicomp}
P.~Valiant.
\newblock Testing symmetric properties of distributions.
\newblock {\em SIAM J. Comput.}, 40(6):1927--1968, 2011.

\bibitem[VV11]{ValiantValiant:11}
G.~Valiant and P.~Valiant.
\newblock {Estimating the unseen: an $n/\log(n)$-sample estimator for entropy
  and support size, shown optimal via new CLTs}.
\newblock In {\em STOC}, pages 685--694, 2011.

\bibitem[VV14]{VV14}
G.~Valiant and P.~Valiant.
\newblock An automatic inequality prover and instance optimal identity testing.
\newblock In {\em FOCS}, 2014.

\bibitem[Wal09]{Walther09}
G.~Walther.
\newblock Inference and modeling with log-concave distributions.
\newblock {\em Statistical Science}, 24(3):319--327, 2009.

\bibitem[Weg70]{Wegman:70}
E.J. Wegman.
\newblock {Maximum likelihood estimation of a unimodal density. I. and II.}
\newblock {\em Ann. Math. Statist.}, 41:457--471, 2169--2174, 1970.

\bibitem[WWW{\etalchar{+}}05]{Wang05}
X.~Wang, M.~Woodroofe, M.~Walker, M.~Mateo, and E.~Olszewski.
\newblock Estimating dark matter distributions.
\newblock {\em The Astrophysical Journal}, 626:145--158, 2005.

\end{thebibliography}


\appendix

\section*{Appendix: Omitted Proofs} \label{sec:ap}

\section{A Useful Primitive: Testing Uniformity in $L_2$ norm} \label{sec:unif-L2}

In this section, we give an algorithm for uniformity testing with respect to the $L_2$ distance, thereby establishing Theorem~\ref{thm:unif_l2-delta}.
The algorithm Test-Uniformity-$L_2 (q, n, \eps)$ described below
draws $O(\sqrt{n}/\eps^2)$ samples from a distribution $q$ over $[n]$ and distinguishes between the cases that $q = U_n$ versus $\|q-U_n\|_2 > \eps/\sqrt{n}$ with probability at least $2/3$.
Repeating the algorithm $O(\log(1/\delta))$ times and taking the majority answer results in a confidence probability of $1-\delta$, giving the desired algorithm
Test-Uniformity-$L_2 (q, n, \eps, \delta)$ of Theorem~\ref{thm:unif_l2-delta}.

Our estimator is a  variant of Pearson's chi-squared test~\cite{Pearson1900}, and can be viewed as a special case of the recent ``chi-square type'' testers
in~\cite{CDVV14, VV14}. We remark that, as follows from the Cauchy-Schwarz inequality, the same estimator distinguishes the uniform distribution
from any distribution $q$ such that $\|q-U_n\|_1 > \eps$, i.e., algorithm  Test-Uniformity-$L_2 (q, n, \eps)$ is an optimal uniformity tester for the $L_1$ norm.
The $L_2$ guarantee we prove here is \new{new}, is strictly stronger \new{than the aforementioned $L_1$ guarantee,} and is crucial for our purposes in Section~\ref{sec:main}.

For $\lambda \ge 0$, we denote by $\Poi(\lambda)$ the Poisson distribution with parameter
$\lambda.$ In our algorithm below, we employ the standard ``Poissonization'' approach: namely, we assume that, rather than drawing $m$ independent samples from a distribution,  we first select $m'$ from $\Poi(m)$, and then draw $m'$ samples.
This Poissonization makes the number of times different elements occur in the sample independent,
with the distribution of the number of occurrences of the $i$-th domain element
distributed as $\Poi(mq_i)$, simplifying the analysis.  As $\Poi(m)$ is tightly concentrated about $m$, we can carry out this Poissonization trick without loss of generality at the expense of only sub-constant factors in the sample complexity.

\medskip
\fbox{\parbox{6in}{

{\bf Algorithm} Test-Uniformity-$L_2 (q, n, \eps)$ \\
Input: sample access to a distribution $q$ over $[n]$, and $\eps>0$.\\
Output: ``YES'' if $q = U_n$; ``NO'' if $\|q-U_n\|_2 \ge \eps/\sqrt{n}.$
\begin{enumerate}
  \item Draw $m' \sim \Poi(m)$ iid samples from $q$.
  \item Let $X_i$ be the number of occurrences of the $i$th domain elements in the sample from $q$
  \item Define $Z=\sum_{i=1}^n(X_i-m/n)^2-X_i.$
  \item If $Z \ge 4m/\sqrt{n}$ return ``NO''; otherwise, return ``YES''.
\end{enumerate}
}}
\medskip

The following theorem characterizes the performance of the above estimator:

\begin{theorem} \label{thm:l2-unif}
For any distribution $q$ over $[n]$ the above algorithm distinguishes the case that $q = U_n$ from the case that $||q-U_n||_2 \ge \eps/\sqrt{n}$ when given $O(\sqrt{n}/\eps^2)$ samples from $q$  with probability at least $2/3$.
\end{theorem}

\begin{proof}


Define $Z_i = (X_i-m/n)^2-X_i$. Since $X_i$ is distributed as $\Poi(mq_i)$, $\E[Z_i] = m^2 \Delta_i^2$, where $\Delta_i := 1/n-q_i.$
By linearity of expectation we can write $\E[Z] = \sum_{i=1}^n \E[Z_i] = m^2 \cdot \sum_{i=1}^n \Delta_i^2.$ Similarly we can calculate
$$ \var[Z_i] = 2m^2(\Delta_i - 1/n)^2 + 4m^3(1/n - \Delta_i)\Delta_i^2.$$ Since the $X_i$'s (and hence the $Z_i$'s) are independent, it follows that
$\var[Z] = \sum_{i=1}^n \var[Z_i].$

We start by establishing completeness.
Suppose $q = U_n$. We will show that $\Pr [Z \ge 4m/\sqrt{n}] \le 1/3.$
Note that in this case $\Delta_i = 0$ for all $i \in [n]$, hence $ \E[Z] = 0$ and
$ \var[Z] = 2m^2/n.$ Chebyshev's inequality implies that
$$ \Pr [Z \ge 4m/\sqrt{n}] = \Pr\left[ Z \ge (2\sqrt{2}) \sqrt{\var[Z]}\right] \le (1/8) < 2/3$$
as desired.

We now proceed to prove soundness of the tester.
Suppose that  $\|q - U_n \|_2 \ge \frac{\eps}{\sqrt{n}}$.
In this case we will show that $\Pr [Z \le 4m/\sqrt{n}] \le 1/3.$
Note that Chebyshev's inequality implies that
$$ \Pr \left[Z \le \E[Z] - 2\sqrt{\var[Z]} \right] \le 1/4.$$
It thus suffices to show that  $\E[Z] \ge 8m/\sqrt{n}$ and $\E[Z]^2 \ge 16\Var[Z].$
Establishing the former inequality is easy. Indeed,
$$\E[Z] = m^2 \cdot \|q-U_n\|_2^2 \ge m^2 \cdot (\eps^2/n) \ge 8m/\sqrt{n}$$
for $m \ge 8\sqrt{n}/\eps^2.$

Proving the latter inequality requires a more detailed analysis. We will show that
for a sufficiently large constant $C>0$, if $m \ge C \sqrt{n}/\eps^2$ we will have
$$ \Var[Z] \ll \E[Z]^2.$$
Ignoring multiplicative constant factors, we equivalently need to show that
$$m^2 \cdot \left( \sum_{i=1}^n \left(\Delta_i^2 - 2\Delta_i/n\right)+1/n \right)  + m^3 \sum_{i=1}^n \left(\Delta_i^2/n + \Delta_i^3\right) \ll m^4 \left( \sum_{i=1}^n \Delta_i^2 \right)^2.$$
To prove the desired inequality, it suffices to bound from above the absolute value of each of the five terms of the LHS separately.
For the first term we need to show that
$ m^2 \cdot \sum_{i=1}^n \Delta_i^2 \ll m^4 \cdot \left( \sum_{i=1}^n \Delta_i^2 \right)^2$ or equivalently
\begin{equation} \label{eqn:one}
m \gg 1/\|q-U_n\|_2.
\end{equation}
Since $\|q-U_n\|_2 \ge \eps/\sqrt{n}$, the RHS of (\ref{eqn:one}) is bounded from above by $\sqrt{n}/\eps$,
hence (\ref{eqn:one}) holds true for our choice of $m$.

For the second term we want to show that
$\sum_{i=1}^n |\Delta_i|  \ll m^2 n \cdot \left( \sum_{i=1}^n \Delta_i^2 \right)^2$. Recalling that
$ \sum_{i=1}^n |\Delta_i|  \le \sqrt{n} \cdot \sqrt{ \sum_{i=1}^n \Delta_i^2}$, as follows from the Cauchy-Schwarz inequality,
it suffices to show that $m^2 \gg (1/\sqrt{n}) \cdot 1/(\sum_{i=1}^n \Delta_i^2)^{3/2}$ or equivalently
\begin{equation} \label{eqn:two}
m \gg \frac{1}{n^{1/4}} \cdot \frac{1}{\|q-U_n\|^{3/2}_2}.
\end{equation}
Since $\|q-U_n\|_2 \ge \eps/\sqrt{n}$, the RHS of (\ref{eqn:two}) is bounded from above by $\sqrt{n}/\eps^{3/2}$,
hence (\ref{eqn:two}) is also satisfied.


For the third term we want to argue that $m^2/n \ll m^4 \cdot \left( \sum_{i=1}^n \Delta_i^2 \right)^2$ or
\begin{equation} \label{eqn:three}
m \gg \frac{1}{n^{1/2}} \cdot \frac{1}{\|q-U_n\|^{2}_2},
\end{equation}
which holds for our choice of $m$, since the RHS is bounded from above by $\sqrt{n}/\eps^2$.

Bounding the fourth term amounts to showing that $ (m^3/n) \sum_{i=1}^n \Delta_i^2  \ll m^4 \left( \sum_{i=1}^n \Delta_i^2 \right)^2$
which can be rewritten as
\begin{equation} \label{eqn:four}
m \gg \frac{1}{n} \cdot \frac{1}{\|q-U_n\|^{2}_2},
\end{equation}
and is satisfied since the RHS is at most $1/\eps^2.$


Finally, for the fifth term we want to prove that
$m^3 \cdot \sum_{i=1}^n |\Delta_i|^3 \ll m^4 \cdot \left( \sum_{i=1}^n \Delta_i^2 \right)^2$ or
that $\sum_{i=1}^n |\Delta_i|^3 \ll m \cdot \left( \sum_{i=1}^n \Delta_i^2 \right)^2.$
From Jensen's inequality it follows that $\sum_{i=1}^n |\Delta_i|^3 \le \left( \sum_{i=1}^n \Delta_i|^2 \right)^{3/2}$;
hence, it is sufficient to show that $\left( \sum_{i=1}^n \Delta_i|^2 \right)^{3/2}   \ll m \cdot \left( \sum_{i=1}^n \Delta_i^2 \right)^2$
or
\begin{equation} \label{eqn:five}
m \gg \frac{1}{\|q-U_n\|_{2}}.
\end{equation}
Since $\|q-U_n\|_{2} \ge \eps/\sqrt{n}$ the above RHS is at most $\sqrt{n}/\eps$ and (\ref{eqn:five}) is satisfied.
This completes the soundness proof and the proof of Theorem~\ref{thm:l2-unif}.
\end{proof}

\end{document}